\documentclass{llncs}

\usepackage{amsfonts}
\usepackage[dvips]{graphicx}
\usepackage{amsmath}
\usepackage{amssymb}
\usepackage{xspace}
\usepackage{colortbl}
\usepackage{mathtools}
\usepackage[utf8]{inputenc}
\usepackage[english]{babel}
\usepackage{fancyhdr}
\usepackage{tabularx}
\usepackage{rotating}
\usepackage{comment}
\usepackage{color}
\usepackage{esvect}
\usepackage[T1]{fontenc} 
\usepackage{stmaryrd}
\usepackage{xcolor}
\usepackage{hyperref}
\usepackage{tikz}
\usetikzlibrary{quantikz}
\usepackage[ruled,lined]{algorithm2e}
\usepackage{lineno}
\usepackage{color}

\newtheorem{observation}{Observation}

\renewcommand{\leq}{\leqslant}
\renewcommand{\geq}{\geqslant}

\usepackage{url}

\newcommand{\COMMENT}[1]{}

\title{Quantum Circuits for\\Fixed Substring Matching Problems\thanks{This work is funded by by ICSC – Centro Nazionale di Ricerca in High-Performance Computing, Big Data and Quantum Computing.}}

\author{Domenico Cantone$^{\dag}$, Simone Faro$^{\dag}$, Arianna Pavone$^{\ddag}$ and Caterina Viola$^{\dag}$}
\institute{$^{\dag}$Department of Mathematics and Computer Science, University of Catania\\
\email{\{domenico.cantone,simone.faro,caterina.viola\}@unict.it}\\[0.1cm]
$^{\ddag}$Department of Mathematics and Computer Science, University of Palermo\\
\email{ariannamaria.pavone@unipa.it}
}

\begin{document}


\maketitle

\newcommand{\pBlock}[1]{p\llbracket #1 \rrbracket}

\begin{abstract}
Quantum computation represents a computational paradigm whose distinctive attributes confer the ability to devise algorithms with asymptotic performance levels significantly superior to those achievable via classical computation. Recent strides have been taken to apply this computational framework in tackling and resolving various issues related to text processing. The resultant solutions demonstrate marked advantages over their classical counterparts. This study employs quantum computation to efficaciously surmount text processing challenges, particularly those involving string comparison. The focus is on the alignment of fixed-length substrings within two input strings. 
Specifically, given two input strings, $x$ and $y$, both of length $n$, and a value $d\leq n$, we want to verify the following conditions: the existence of a common prefix of length $d$, the presence of a common substring of length $d$ beginning at position $j$ (with $0\leq j<n$) and, the presence of any common substring of length $d$ beginning in both strings at the same position. Such problems find applications as sub-procedures in a variety of problems concerning text processing and sequence analysis.
Notably, our approach furnishes polylogarithmic solutions, a stark contrast to the linear complexity inherent in the best classical alternatives.
\end{abstract}

\section{Introduction}


Quantum computing stands as a swiftly evolving domain within the realm of computer science, harnessing the tenets of quantum mechanics to forge computing systems of heightened potency, diverging significantly from the traditional modus operandi of classical computers.

In stark contrast to classical counterparts, which hew to the binary realm of $0$s and $1$s for data processing, quantum computing capitalizes on qubits, entities capable of simultaneous existence in multiple states. Furthermore, the enigmatic phenomenon of quantum entanglement comes into play, permitting the orchestration of parallel operations by two or more qubits. This synergy of qubits facilitates swifter, more efficient operations than their classical bit counterparts, unfurling a distinctive edge for quantum computers. This advantage is particularly conspicuous in arenas such as cryptographic breaking and optimization, where quantum prowess delivers extraordinary computational velocities.

The impact of quantum computing reverberates through the algorithmic landscape, epitomized by landmark strides like Shor's algorithm~\cite{Shor1997} for decomposing large numbers and Grover's algorithm~\cite{Grover96} for unordered searches. These algorithms confer exponential and quadratic accelerations over classical counterparts, respectively, standing as compelling testament to the prowess of quantum computing. Their emergence not only underscores the immense potential of this field but also kindles an effusion of fervor for amplified exploration and advancement.
Nonetheless, it is the recent substantiation of quantum supremacy that has galvanized an upsurge of fascination around quantum computing, catalyzing the infusion of these novel technologies across diverse domains within the realm of computer science.

Only in recent times, the realm of quantum computation has turned its gaze towards the domain of text processing and string matching, as evidenced by seminal works such as \cite{RAMESH2003103,Niroula21,CF23,0022403}.
In the realm of string matching problems, a pioneering quantum algorithm emerged in the form of Ramesh and Vinay's $\tilde{\mathcal{O}}(\sqrt{n} + \sqrt{m})$ time algorithm, as outlined in~\cite{RAMESH2003103}, facilitating the identification of the initial exact occurrence of a pattern with a length of $m$ within a text of length $n$. Later, Niroula and Nam introduced an algorithm \cite{Niroula21} for exact string matching, boasting a time complexity of $\tilde{\mathcal{O}}(\sqrt{n})$. Further advancements  \cite{CF23} have since extended the Niroula-Nam quantum algorithm to address swap matching in $\tilde{\mathcal{O}}(\sqrt{n2^{m}})$ time, thus culminating in a time complexity of $\tilde{\mathcal{O}}(n)$ in cases where the pattern length is short, i.e., when $m=\mathcal{O}(\log(n))$.
There have been other quantum solutions to string problems, mostly focused on the edit distance~\cite{0022403}. 

In this paper we are concerned with solving some combinatorial problems related to string comparison. Specifically, given two input strings, $x$ and $y$, whose characters are drawn from an alphabet $\Sigma$ of size $\sigma$, both of length $n$, and a parameter $d$, with $0< d \leq n$, we tackle the following problems: 
\begin{enumerate}
\item the Fixed Prefix Matching (FPM) problem, which checks for the presence of a common prefix between $x$ and $y$ of length $d$, i.e., checks whether 
$$x[0 .. d-1] = y[0 .. d-1]$$
\item the Fixed Factor Matching (FFM) problem, which tests, for a given input parameter $j$, with $0\leq j < n-d$, for the existence of a common factor between $x$ and $y$, beginning at position $j$ of both strings, i.e. it tests whether 
$$x[j .. j+d-1] = y[j .. j+d-1]$$ 
\item the Shared Fixed Substring Checking (SFSC) problem, which tests for the existence of a common factor between $x$ and $y$ with beginning at the same position in both strings, i.e. it tests whether 
$$\exists\ j\ :\ 0\leq j <n-d \textrm{ and } x[j .. j+d-1] = y[j .. j+d-1]$$
\end{enumerate}

We call this family of problems by \emph{fixed substring matching} (FSM) problems, since, given the parameter $d\leq n$, we are interested in matching substrings of \emph{fixed} length $d$. Although it is possible to identify other problems belonging to this family, they can all be traced back to one of the three problems listed above.
Note also that the FPM problem can be solved by any solution for the FFM problem by posing $j=0$.

These are three basic problems in many applications of text processing, computational biology, and sequence analysis in general. Solutions to these problems can be used as building-blocks for solving more complex text processing problems such as longest common substring~\cite{cormen01introduction,GallS23}, longest palindromic substring, sequence alignments or string-distance computation problems, just to list a few. Their resolution is therefore of fundamental importance in such areas.

Using the classical computation paradigm, these three problems can be solved in linear time, including through the use of simple data structures built on the two input strings, such as generalized suffix trees~\cite{323593}. In this paper, we propose solutions based on quantum computation that allow us to obtain, in all three cases, a solution in polylogarithmic time. Specifically, we propose an approach based on the circuit-based computational model that is well suited for solving all three problems addressed, and that allows obtaining a solution in $\mathcal{O}(\log^3(n))$ computational time when processing binary strings, while obtaining a complexity equal to $\mathcal{O}(\log^4(n))$ in the general case.

The paper is organized as follows.  In Section \ref{sec:preliminaries}, we review some useful preliminaries.
in Section \ref{sec:general} we present a general procedural solution for solving our problems. Finally in Section \ref{sec:actual_circuits} we give an actual implementation by quantum circuits of the proposed algorithms by providing their structural details and a complexity analysis.

\section{Preliminaries}\label{sec:preliminaries}

Given a string $x$, of length $n \geq 1$, over a finite alphabet $\Sigma$ of size $\sigma$, we represent it as a finite array $x[0\,..\,n-1]$. 
The empty string is denoted by $\varepsilon$.
We denote by $x[i]$ the $(i+1)$-st character of $x$, for $0\leq i< n$, and by $x[i\,..\,j]$ the substring of $x$ contained between the $(i+1)$-st and the $(j+1)$-st characters of $x$, for $0\leq i \leq j <n$.  
A $k$-substring of a string $x$ is a substring of $x$ of length $k$. 
For ease of notation, the $(i+1)$-st character of the string $x$ will also be denoted by the symbol $x_i$, so that $x = x_0 x_1 \ldots x_{n-1}$.   
A substring of $x$ beginning at position $0$ is a \emph{prefix} of $x$.  We use the symbol $x_{:i}$ to indicate the prefix of $x$ of length $i$.

For any two strings $x$ and $y$ of length $n$, we say that they have a common $k$-substring  if two values $i$ and $j$ exist, such that $0\leq i,j<n-k$ and $x[i\ldots i+k-1] = y[j\ldots j+k-1]$. For our convenience, we instead say that $x$ and $y$ \emph{share} a $k$-substring if a position $i$ exists such that $x[i\ldots i+k-1] = y[i\ldots i+k-1]$. We write $x\cdot y$ to denote the concatenation of $x$ and $y$. Furthermore, given a string $x$ of length $n$ and a shift value $0\leq j <n$, we denote by the symbol $\vv{x}^j$ the cyclic rightward rotation of the characters of $x$ by $j$ positions. More formally, we have $\vv{x}^j = x[n-j\ldots n-1]\cdot x[0\ldots n-j-1]$.

\subsection{Essential Basics of Quantum Computation}\label{sec:basics}
The fundamental unit in quantum computation is the \emph{qubit}.
A qubit is a coherent superposition of the two orthonormal basis states, which are denoted by $|0\rangle$ and $|1\rangle$, using the conventional \emph{bra–ket} notation. The mathematical formulation for a qubit $|\psi\rangle$ is then as a linear combination of the two basis states, i.e., $|\psi\rangle = \alpha |0\rangle + \beta |1 \rangle$, where the values $\alpha$ and $\beta$, called \emph{amplitudes}, are complex numbers such that $|\alpha|^2 + |\beta|^2 = 1$, representing the probability of finding the qubit in either the state $|0\rangle$ or $|1\rangle$, respectively, when measured. A quantum measurement is the only operation through which information can be gained about the state of a qubit, however causing the qubit to collapse to one of the two basis states. The measurement of the state of a qubit is irreversible, meaning that it irreversibly alters the magnitudes of $\alpha$ and $\beta$. If $b$ is a binary value, equal to $0$ or $1$, we use the symbol $|b\rangle$ to indicate the qubit in the corresponding basis state, $|0\rangle$ or $|1\rangle$, respectively.
Multiple qubits taken together are referred to as \emph{quantum registers}. A quantum register $|\psi\rangle = |q_0, q_1, .., q_{n-1}\rangle$ of $n$ qubits is the tensor product of the constituent qubits, i.e., $|\psi\rangle = \bigotimes_{i=0}^{n-1}|q_i\rangle$.
If $k$ is an integer value that can be represented by a binary string of length $n$, we use the symbol $|k\rangle$ to denote the register of $n$ qubits such that $|k\rangle = \bigotimes_{i=0}^{n-1}|k_i\rangle$, where $|k_i\rangle$ takes the value of the $i$-th least significant binary digit of $k$. For example, the quantum register $|9\rangle$ with $4$ qubits is given by $|9\rangle = |1\rangle \otimes |0\rangle \otimes |0\rangle \otimes |1\rangle$.
The mathematical formulation of a a quantum register is then $|\psi\rangle = \sum_{k=0}^{2^n-1} \alpha_k |k\rangle$, where the values $\alpha_k$ represent the probability of finding the register in the state $|k\rangle$  when measured, with $\sum_{k=0}^{2^n-1}|\alpha_k|^2 = 1$.

\emph{Operators} in quantum computing are mathematical entities used to represent functional processes that result in the change of the state of a quantum register. 
Although there is no problem in realizing any quantum operator capable of working in constant time on a quantum register of fixed size, operators of variable size can only be implemented through the composition of elementary gates.

Given a function $f\colon \{0,1\}^n \rightarrow \{0,1\}$, any quantum operator that  maps a register containing the value of a given input $x \in \{0,1\}^n$ into a register whose value depends on $f(x)$ is called a \emph{quantum oracle}.  
A \emph{Boolean oracle} $U_f$ maps a register $|\psi\rangle$ of size $n+1$ initialized to $|x\rangle \oplus |0\rangle$, to the register $|x\rangle \oplus |f(x)\rangle$. More formally,  $U_f |x,0\rangle = |x,f(x)\rangle$.
A \emph{phase oracle} $P_f$ for a function $f\colon \{0,1\}^n \rightarrow \{0,1\}$ takes as input a quantum register $x\in \{0,1\}^n$ and leaves its value unchanged, applying to it a negative global phase only in the case in which $f(x)=1$, that is, only if $x$ is a solution for the function. More formally,  $P_f |x\rangle = (-1)^{f(x)}|x\rangle$.

A quantum algorithm can be shaped using mainly two different models, the query based model and the circuit computational model.

The \emph{quantum query model} takes a macroscopic perspective in quantum computing. It's concerned with quantifying the minimum number of times an algorithm needs to interact with the input to solve a problem efficiently. This model abstracts the intricate details of quantum gates and circuits and centers on the notion of query complexity. It's particularly useful for problems where classical solutions involve many queries and quantum algorithms promise significant improvements in efficiency.

In contrast, the \emph{quantum circuit computational model} offers a micro-level view of quantum computing. It entails representing quantum algorithms as sequences of quantum gates that manipulate qubits. This approach delves into the specifics of how quantum operations are executed and provides a structured methodology for designing and optimizing quantum algorithms. Quantum circuit diagrams visually illustrate the flow of gate operations and the connections between qubits.

In this paper we adopt the circuit model of computation. We assume the circuit as being divided into a sequence of discrete time-steps, where the application of a single gate requires a single time-step. The \emph{depth} of the circuit is the total number of required time-steps and it is the time complexity measure that we adopt in this paper. We remark that the depth of a circuit is not necessarily the same as the total number of gates (called \emph{size}) of the circuit itself, since gates that act on disjoint qubits can often be applied in parallel.

We use the notations $\tilde{\mathcal{O}}(\cdot)$ and $\tilde{\Omega}(\cdot)$ that hide the polylogarithmic factors in terms of $n$.
In addition, we say that a quantum algorithm solves a problem if it finds a solution with probability at least $9/10$. This success probability can be easily amplified to $1-1/$poly$(n)$ with a logarithmic overhead in the complexity.

\subsection{Elementary Quantum Operators}
There is a variety of quantum operators capable of operating on quantum registers to perform widely ranging manipulations. Here, we focus on key components essential for building a quantum circuit, with special focus on those that are useful for this paper's purposes,  highlighting their depth and computational complexity.

The \emph{Pauli}-X (or $X$ or NOT) gate is the quantum equivalent of the NOT gate for classical computers with respect to the standard basis $|0\rangle$ and $|1\rangle$. It operates on a single qubit, mapping $|0\rangle$ to $|1\rangle$ and $|1\rangle$ to $|0\rangle$.

The \emph{Hadamard} (or $H$) gate is a well-known single-qubit operation that maps the basis states $|0\rangle$ and $|1\rangle$ to $\frac{1}{\sqrt{2}}(|0\rangle + |1\rangle)$ and to $\frac{1}{\sqrt{2}}(|0\rangle - |1\rangle)$, respectively, thus creating  a superposition of the two basis states with equal amplitudes.

The \emph{Pauli}-Z (or $Z$ or \emph{phase-flip gate}) is the third single-qubit operator of our interest. It leaves the basis state $|0\rangle$ unchanged while mapping $|1\rangle$ to $-|1\rangle$, by applying a negative phase to it. Based on the equivalence $Z = HXH$, the Z operator can be obtained from the previous two operators.

The \emph{controlled} NOT \emph{gate} (or CNOT) is a quantum logic gate operating on a register of two qubits $|q_0,q_1\rangle$. If the control qubit $|q_0\rangle$ is set to 1, it inverts the target qubit $|q_1\rangle$, otherwise all qubits stay the same. Formally, it maps $|q_0,q_1\rangle$ to $|q_0, q_0\oplus q_1\rangle$.

The \emph{Toffoli gate} (also known as CCNOT gate) is a universal reversible logic gate that works on $3$ qubits: if the first two qubits are both set to 1, it inverts the third qubit, otherwise all bits stay the same. Formally, it maps a 3 qubits register $|q_0, q_1, q_2\rangle$ to $|q_0, q_1, q_0q_1\oplus q_2\rangle$.

Generalizations of the CNOT gate are the $n$-ary \emph{fanout} operator and the \emph{multiple}-CNOT. 
Given a quantum register $|\psi\rangle = |q_0, q_1, \ldots, q_{n-1}\rangle$ of $n$ qubits, a \emph{fanout} operator simultaneously copies the control qubit $|q_0\rangle$ onto the $n-1$ target qubits $|q_i\rangle$, for $i=1, \ldots, n-1$. 
Formally the fanout operator applies the following mapping 
$|q_0,q_1,q_2, \ldots, q_{n-1}\rangle$ to $|q_0, q_0\oplus q_1, q_0\oplus q_2, \ldots, q_0\oplus q_{n-1}\rangle$.
Although a constant time fanout can be obtained by the product of $n$ controlled-not gates, the no-cloning theorem makes it difficult to directly fanout qubits in constant depth~\cite{Hyer2005QuantumFI}. However, assuming that the target qubits are all initialized to $|0\rangle$, it is easy to see that, by a  divide-and-conquer strategy, we can compute fanout in depth $\Theta(\log(n))$ using controlled-not gates and $0$ ancill\ae\ qubits~\cite{Fang06}.

A \emph{multiple}-CNOT (or M-CNOT) operator flips the unique target qubit $|q_{n-1}\rangle$ if all the $n-1$ control qubits $|q_i\rangle$ (for $i=0, .., n-2$) are set to $|1\rangle$. 
Formally, the M-CNOT operator maps 
$|q_0,q_1,q_2, .., q_{n-1}\rangle$  to $|q_0, q_1, \ldots,  q_{n-2}, (q_0 \cdot q_1 \cdots q_{n-2}) \oplus q_{n-1}\rangle$.
The most direct way to implement an M-CNOT operator is to use the concepts of classical Boolean logic as shown in~\cite{Toffoli80} for classical circuits, and later in~\cite{barenco1995} in the field of quantum computing. However, this implementation has a linear depth with respect to the number of control qubits and requires $n-2$ ancill\ae\ qubits. In more recent papers~\cite{Yong17,Balauca22}, it was shown that a construction rearranging the circuit gates to achieve logarithmic depth by exploiting parallelism is possible, using $n-2$ and $n/2$ ancill\ae\ qubits, respectively. 

For the sake of completeness, we also mention a recent result~\cite{Rasmussen20} that enables the implementation of multi-controlled NOT gates in constant time in architectures with trapped ions and neutral atoms.

The \textsf{Swap} \emph{gate} is a two-qubit operator: expressed in basis states, it swaps the state of the two qubits $|q_0,q_1\rangle$ involved in the operation, mapping them to $|q_1,q_0\rangle$.
Finally the \emph{Fredkin gate} (or \emph{controlled-}\textsf{Swap} \emph{gate}) is a universal gate that consists in a \textsf{Swap} gate controlled by a single qubit. Both \textsf{Swap} and Fredkin gates can be obtained from the application of 3 CNOT gates and 3 Toffoli gates, respectively. 

Often in a quantum circuit it is necessary to apply  an operator controlled by a specific qubit $|c\rangle$. 
If we assume that the operator consists exclusively of CNOTs and Toffoli gates, it is not difficult to construct a controlled version of it, where all CNOT gates are converted into Toffoli gates, by adding a new control qubit in $|c\rangle$, while each Toffoli gate is converted into a Multiple-CNOT gate with 3 controls. The resulting controlled operator has the same depth as the uncontrolled operator, except for a higher proportionality factor.\\

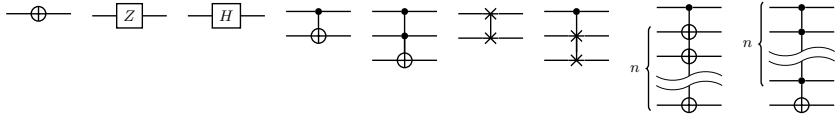
\begin{figure}[!t]
	\begin{tikzpicture}
	\node[scale=0.65] {
		\begin{quantikz}[row sep={0.5cm,between origins}]
		& \targ{} & \qw\\
		&&\\
		&&\\
		&&\\
		&&
		\end{quantikz}\ \ 
		\begin{quantikz}[row sep={0.5cm,between origins}]
		& \gate{Z} & \qw\\
		&&\\
		&&\\
		&&\\
		&&
		\end{quantikz}\ \ 
		\begin{quantikz}[row sep={0.5cm,between origins}]
		& \gate{H} & \qw\\
		&&\\
		&&\\
		&&\\
		&&
		\end{quantikz}\ \
		\begin{quantikz}[row sep={0.5cm,between origins}]
		& \ctrl{1} & \qw \\
		& \targ{} & \qw \\
		&&\\
		&&\\
		&&
		\end{quantikz}\ \
		\begin{quantikz}[row sep={0.5cm,between origins}]
		& \ctrl{2} & \qw \\
		& \ctrl{1} & \qw \\
		& \targ{}  & \qw \\
		&&\\
		&&
		\end{quantikz}\ \
		\begin{quantikz}[row sep={0.5cm,between origins}]
		& \targX{} & \qw \\
		& \swap{-1} & \qw \\
		&&\\
		&&\\
		&&
		\end{quantikz}\ \ 
		\begin{quantikz}[row sep={0.5cm,between origins}]
		& \ctrl{2} & \qw \\
		& \targX{} & \qw \\
		& \swap{-1} & \qw \\
		&&\\
		&&
		\end{quantikz}\ \ 
		\begin{quantikz}[row sep={0.5cm,between origins}]
		& \ctrl{4} & \qw \\
		\lstick[wires=4]{$n$} & \targ{}  & \qw \\
		& \targ{}  & \qw \\
		\wave &&\\
		& \targ{}  & \qw
		\end{quantikz}\ \ 
		\begin{quantikz}[row sep={0.5cm,between origins}]
		\lstick[wires=4]{$n$} & \ctrl{1} & \qw \\
		& \ctrl{2}  & \qw \\
		\wave &&\\
		& \ctrl{1}  & \qw \\
		& \targ{}  & \qw
		\end{quantikz}};
	\end{tikzpicture}
	\caption{\label{fig:basic-gates}The representation of the following basic gates (from left to right): Pauli-X, Pauli-Z, Hadamard, CNOT, Toffoli, Swap, Fredkin, Fanout, and MCNOT.}
\end{figure}

\section{A Common Approach to Fixed Matching Problems}\label{sec:general}

In this section we provide a procedural solution for the three problems of interest to us. 
However, we first introduce the main tools useful in solving our problems on strings. We also give some definitions useful for later discussion and suggest an implementation of the main functions that will later serve as building blocks for the quantum algorithms we will introduce in later sections. When necessary we will prove our claims by means of technical lemmas.

Before entering into the details it is necessary to formalize some assumptions we make along the description.

Since a quantum register of dimension $\log(n)$ can take on all values between $0$ and $n-1$, like any binary sequence of the same dimension, for simplicity we will assume that both input strings $x$ and $y$ have length $n=2^p$, for some $p>0$. We also assume that $x$ and $y$ end with two different special characters, \$ and \%,  not belonging to the alphabet $\Sigma$. These assumptions can be made without loss of generality since it would suffice to take the smallest value $p$ for which we have $n<2^p$ and concatenate the text with $2^p-n$ copies of the special character. 
For instance if $x=abaacbcbbca$ is a text of length $11$, we silently concatenate it with $5$ copies of the character \$, i.e., we assume that $x=abaacbcbbca\$\$\$\$\$$, namely a string of length $16$. This assumption has no effect on the complexity since the resulting string has at most twice the length of the original string.

We also observe that any substring of length $d$ can begin at any position $j$ of the text, for $0\leq j\leq n-d$. However, in this paper we also relax such condition by admitting values of $j$ between $0$ and $n-1$ and assuming that a substring of the text can be obtained in a circular way. This assumption can also be made without loss of generality since the last character of $x$ is always the special character \$, while the last character of $y$ is always the special character \%. This ensures us that no circularly obtained substring can be returned as LCS.%

\begin{observation}\label{obs:binary}
For simplicity of discussion we present the circuits and the resulting algorithms as processing binary strings. This simplification, however, does not lead to any substantial change in our result since, assuming that each character can be represented with (at most) $\log(n)$ bits, it is easy to show that the quantum operators used in the construction of the algorithm would undergo an increase in their complexity at most equal to a factor of $\log(n)$.
Thus, henceforth we assume that $x$ and $y$ are two binary strings, both of length $n$, and further assume that $d$ is a positive value with $d\leq n$.
\end{observation}

\subsection{Useful Tools for Substring Matching}
The first definition we give is that of \emph{matching substring vectors}.
Given a value $i\geq 0$ we define the matching substring vectors $\lambda^i$, for $0\leq i \leq \log(d)$, as bit-vectors of size $n$ such that, for $0\leq j < n-$, we have

$$
	\lambda^i[j] = \begin{cases}
			1 & \textrm{ if } x[j \,..\, j+2^i-1] =  y[j \,..\, j+2^i-1]\\[-.1cm]
			0 & \textrm{ otherwise.}
 		\end{cases}
$$

Roughly speaking, $\lambda^i$ contains a bit set at position $j$ only if the substring of $x$, of length $2^i$ and  beginning at position $j$, matches its counterpart in $y$.

Such vectors can be recursively computed by means of the following relation
\begin{equation}\label{eq:varphi}
	\lambda^i[j] = \begin{cases}
			1 & \textrm{ if } i=0 \textrm{ and } x[j] =  y[j]\\[-.1cm]
			0 & \textrm{ if } i=0 \textrm{ and } x[j] \neq y[j]\\[-.1cm]
			\lambda^{i-1}[j] \textrm{ \textsc{and} } \lambda^{i-1}[j+2^{i-1}] & \textrm{ if } i>0.
 	\end{cases}
\end{equation}

\begin{example}
Let $x=\texttt{cgaactta}$ and $y=\texttt{ctacctta}$ be two strings of length $8$. Regarding the values of the vectors $\lambda^i$ we observe the following configuration
$$
\begin{array}{lll}
	\lambda^0 & = \texttt{ [ 1 0 1 0 1 1 1 1 ] } & \longleftarrow\ x[i] = y[i] \textrm{ for } i\ \in\ \{0,2,4,5,6,7\}\\
	\lambda^1 & = \texttt{ [ 0 0 0 0 1 1 1 0 ] } & \longleftarrow\ x[i .. i+1] = y[i .. i+1] \textrm{ for } i\ \in\ \{4,5,6\}\\
	\lambda^2 & = \texttt{ [ 0 0 0 0 1 0 0 0 ] } & \longleftarrow\ x[i .. i+3] = y[i .. i+3] \textrm{ for } i\ \in\ \{4\}\\
	\lambda^3 & = \texttt{ [ 0 0 0 0 0 0 0 0 ] } & \\
\end{array}
$$
\end{example}

\smallskip

Let  $\bar{d} = \bar{d}[0\,..\,\lceil \log(d) \rceil]$ be the vector storing the binary representation of $d$ in its reverse order, i.e., $d = \sum_{i=0}^{\log(d)} \bar{d}[i]2^i$. We use the expression $\bar{d}[0{:}m]$ to denote the partial sum of such binary representation, up to $m+1$ bits. Formally we have $\bar{d}[0{:}m] = \sum_{i=0}^{m} \bar{d}[i]2^i$. We set by default $\bar{d}[0{:}i] = 0$ for any $i<0$.

\smallskip

In addition, let $S$ be the sequence of positions $i$ for which $\bar{d}[i]=1$, in increasing order, so that $d = \sum_{i\in S} 2^i$. We say that a string $w$ of length $d$ can then be decomposed into the concatenation of $|S|$ substrings so as to have
\begin{equation}\label{EQ:lemma3}
	w = \underset{i \in S}{\biguplus}\, w\Big[\bar{d}[0{:}i-1]\, ..\, \bar{d}[0{:}i]-1 \Big] ,
\end{equation}
where we used the symbol $\underset{i \in S}{\biguplus}$ to denote string concatenation over the ordered sequence $S$ of positions $i$.
Observe that the length of the substring $w[\bar{d}[0{:}i-1] \,..\, \bar{d}[0{:}i]]$ is equal to $\bar{d}[0{:}i]-\bar{d}[0{:}i-1] = \sum_{j=0}^{i-1}2^j - \sum_{j=0}^{i}2^j = 2^i$.
We call such a decomposition of $w$ as its \emph{power-based decomposition}.

\begin{example}
Let $w=\texttt{cgaac}$ be a string of length 5. Since the value $d=5$ is decomposed into $1+4$, the power-based decomposition of $w$ is the sequence of substrings $\langle \texttt{c}, \texttt{gaac} \rangle$. In such a case the set $S$ is equal to $S=\{0,2\}$.
Assume now that $w=\texttt{cgaacttacgt}$ be a string of length 11. Since the value $d=11$ is decomposed into $1+2+8$, the power-based decomposition of $w$ is the sequence of substrings $\langle \texttt{c}, \texttt{ga}, \texttt{acttacgt} \rangle$. In this second case the set $S$ is equal to $S=\{0,1,3\}$.
\end{example}

We state the following technical lemma.

\begin{lemma}\label{lem:bms}
Let $x$ and $y$ be any two strings of length $n$ over the same alphabet $\Sigma$, and let $d$ be a constant, with $0<d\leq n$. 
For a given $j$ such that $0\leq j< n-d$, we have
\begin{equation}\label{eq:thm}
	x[j\,..\,j+d-1]=y[j\,..\,j+d-1]\ \ \iff\ \  \bigwedge_{i \in S}\lambda^i\Big[j+\bar{d}[0{:}i-1]\Big] = 1,
\end{equation}
where $S$ stands for the sequence of all positions $i$ in $bar{d}$ such that $bar{d}[i]=1$, in increasing order.
\end{lemma}
\begin{proof}
If we denote by $\bar{d}$ the vector that stores the binary representation of $d$ and by $S$ the sequence of positions $i$ for which $\bar{d}[i]=1$, in increasing order, then it is straightforward to prove that $d = \sum_{i\in S} 2^i$.

Suppose now that $x[j \,..\, j+d-1] = y[j \,..\, j+d-1]$. Based on equation (\ref{EQ:lemma3}), the substring $x[j \,..\, j+d-1]$ can be decomposed into the substrings $x[j + \bar{d}[0{:}i-1]\, ..\, j + \bar{d}[0{:}i]-1]$, each of which coincides with its counterpart in $y$, for every position $i$ in $S$. Since $x[j + \bar{d}[0{:}i-1]\, .. \, j + \bar{d}[0{:}i]-1] = y[j + \bar{d}[0{:}i-1]\, ..\, j + \bar{d}[0{:}i]-1]$ if and only if $\lambda^i[j + \bar{d}[0{:}i-1]]$, we get the biimplication (\ref{eq:thm}), proving the lemma. 

\qed
\end{proof}

\begin{example}
Let $x=$\texttt{agccatgccaatgcat} and $y=$\texttt{cgcgataccaattcat} two strings of length $16$, and let $d=5$. Observe that $x$ and $y$ have a matching substring of length $5$ beginning at position $7$ of both strings. The vector $\bar{d}$ corresponding to the reverse binary representation of  $d$ is $\bar{d}=$\texttt{101}, thus we need to compute the two bit vectors $\lambda^0=$\texttt{0110110111110111} and $\lambda^2=$\texttt{0000000110000000}. According to Lemma \ref{lem:bms} and regarding the substrings beginning at position $7$, we have $x[7 \,..\, 12] = y[7 \,..\, 12]$ if and only if $\lambda^0[7] \wedge \lambda^2[8] = 1$, which is indeed the case in our example.
\end{example}

\subsection{The Algorithm}
Based on Lemma \ref{lem:bms}, the algorithm maintains the set of vectors $D^i$, for $0\leq i < \log(d)$, which retains information about substrings of length $\bar{d}[0{:}i]$ in common between $x$ and $y$, maintaining a bit set to $1$ at position $j+\bar{d}[0{:}i]$ when a common substring of length $\bar{d}[0{:}i]$ begins at position $j$. 

More formally, for $0\leq i<\log(d)$, we have 

\begin{equation}\label{eq:D}
	D^i[j+\bar{d}[0{:}i]] = 1\ \iff\ x[j \,..\, j + \bar{d}[0,i]-1] = y[j \,..\, j + \bar{d}[0,i]-1].
\end{equation}

The idea is to process strings whose length increases based on the power-based decomposition of all substrings of length $d$.  
Roughly speaking, the substring $x[j \,..\, j+\bar{d}[0,i]-1]$ performs a match with its counterpart in $y$ if and only if $D^i[j+\bar{d}[0{:}i]]=1$.

Since we assume that an empty substring (case where we have $d=0$) starting at position $j$ is always shared between the two strings, for any value $0\leq j < n$, the vector can be initialized to $D^{-1}=1^{n+1}$.

\begin{observation}
Note that to solve the FPM problem we should allow the substring in common to be identified with beginning at position $0$. This suggests that we initialize the vector $D^{-1}$ to the value $10^n$. Differently if we want to identify a possible common string with beginning at position $j$ (FFP problem), the vector $D^{-1}$ should be initialized to the value $0^{j}10^{n-j}$. Finally if we initialize the vector $D^{-1}$ to the value $1^{n+1}$ we allow the common substring to be recognized with beginning at any position $j$, with $0\leq j < n$, solving the SFSC problem.
\end{observation}

After the initialization phase the vector $D^i$, for $0\leq i \leq \log(d)$, can be computed recursively by the following formula

\begin{equation}\label{eq:recD}
	D^i = (D^{i-1} \wedge \lambda^i) \gg 2^i
\end{equation}

where the $\wedge$ operator represents the bit-to-bit logical \textsc{and} between the elements of the two vectors, and the $\gg$ operator represents the rightward shift operator. Such computation allows us to have, after $\log(d)$ steps of the iteration, $x[j \,..\, j+d-1] = y[j \,..\, j+d-1]$ if and only if $D^{\log(d)}[j+d]=1$.

\begin{figure}[!t]
\begin{footnotesize}
\begin{center}
\rotatebox[origin=c]{90}{\scriptsize Step 2~~~~~~~~~~Step 1~~~~~~~~~~Step 0~~~~~~~~~~Input}~~~~~~
\begin{tabular}{lrll}
	\hline
						& & & \\
						& $x =$ 		 								& \texttt{[ a g c c a t g c c a a t g c a t ]} & \\[0.1cm]
						& $y =$ 		 								& \texttt{[ c g c g a t a c c a a t t c a t ]} & \\[0.1cm]
						& $D^{-1} =$ 										& \texttt{[ 1 1 1 1 1 1 1 1 1 1 1 1 1 1 1 1 1 ]} & \\
						& & & \\
	\hline	
						& & & \\
						& $\lambda^0 = \textsc{M}(x,y)$ 								& \texttt{[ 0 1 1 0 1 1 0 1 1 1 1 1 0 1 1 1 ]} &  \\[0.1cm]
	$\bar{d}[0] = 1 \rightarrow$	& $D^0 = (D^{-1} \wedge \lambda^0) \gg 2^0 =$				& \texttt{[ 0 0 1 1 0 1 1 0 1 1 1 1 1 0 1 1 1 ]} & \\
						& & & \\
	\hline
						& & & \\
						& $\lambda^1 = \textsc{ext}1(\lambda^0) = $				& \texttt{[ 0 1 0 0 1 0 0 1 1 1 1 0 0 1 1 ]} & \\[0.1cm]
	$\bar{d}[1] = 0 \rightarrow$	& $D^1 = D^{0} = $				& \texttt{[ 0 0 1 1 0 1 1 0 1 1 1 1 1 0 1 1 1 ]} & \\
						& & & \\
	\hline
						& & & \\
						& $\lambda^2 = \textsc{ext}2(\lambda^1) = $				& \texttt{[ 0 0 0 0 0 0 0 1 1 0 0 0 0 ]} & \\[0.1cm]
	$\bar{d}[2] = 1 \rightarrow$ 	& $D^2 = (D^1 \wedge \lambda^2) \gg 2^2 =$ 				& \texttt{[ 0 0 0 0 0 0 0 0 0 0 0 0 1 0 0 0 0 ]} & \\
						& & & \\
	\hline
\end{tabular}
\end{center}
\end{footnotesize}
\caption{\label{fig:example}An example of the execution of the FSM algorithm on two input strings $x$ and $y$ of size $16$ and a parameter $d$ equal to $5$ representable by $3$ bits ($\bar{d} = [1,0,1]$). The algorithm thus consists of three execution steps.}
\end{figure}

In Fig.\ref{fig:algorithm-bpm-pseudocode} we show a procedural representation of the algorithm for solving the generic problem for the fixed substring matching problem, called FSM, while in Fig.\ref{fig:example} we show an example of the execution of the FSM algorithm on two input strings $x$ and $y$ of size $16$ and a parameter $d$ equal to $5$ representable by $3$ bits ($\bar{d} = [1,0,1]$). The algorithm thus consists of three execution steps. 

The procedure takes as input, in addition to the two strings $x$ and $y$ and the parameter $d$, the vector $D^{-1}$ initialized according to one's needs. The procedure consists of an initialization phase of the vector $\lambda^0$ (line 1) and the vector $D^0$ (line 2), and a main loop that is executed $\log(d)$ times (lines 3-6). The $i$-th iteration step of the loop copies the vectors $\lambda^i$ and $D^i$ by exploiting the proposed recursive relations in (\ref{eq:varphi}) and (\ref{eq:recD}), respectively. At the end of the main loop, the logical \textsc{or} of all elements contained in the vector $D^{\log(d)-1}$ is computed and returned in the output (line 7).

\begin{figure}[!t]
\begin{center}
\begin{footnotesize}
\begin{tabular}{lll}
	\multicolumn{3}{l}{\textsc{FSM}$(d,x,y,D^{-1})$:}\\
	1.& \textsf{$\lambda^0 \leftarrow \textsf{M}(x,y)$} &  \\
	2.& \textsf{if $(\bar{d}[0]=1)$ then $D^0 \leftarrow (D^{-1} \wedge \lambda^0) \gg 2^0$} &  \\
	3.& \textsf{for $i \leftarrow 1$ to $\log(d)$ do } &  \\
	4.& \quad \textsf{$\lambda^i \leftarrow \textsc{ext}_i(\lambda^{i-1})$} &  \\
	5.& \quad \textsf{if $(\bar{d}[i]=1)$ then $D^{i} \leftarrow (D^{i-1} \wedge \lambda^{i}) \gg 2^i$} &  \\
	6.& \quad \textsf{else $D^{i} \leftarrow (D^{i-1}$} &  \\
	7.& \textsf{$r \leftarrow \bigvee_{j=d}^{n+d} D^{\log(d)}[j]$} &  \\
	8.& \textsf{return $r$} &  \\
\end{tabular}
\end{footnotesize}
\end{center}
\vspace{-0.5cm}
\caption{\label{fig:algorithm-bpm-pseudocode}The pseudocode of the FSM algorithm}
\end{figure}

\section{Actual Implementation by Quantum Circuits}\label{sec:actual_circuits}

In this section we provide a real implementation of the functions and operators described above through the definition of specific quantum circuits.

Since a quantum circuit is constructed by the composition of quantum operators, which in turn are constructed by the composition of elementary quantum gates, we first proceed to define the quantum operators useful in solving our problems. We will then use these operators as building blocks to give a simplified schematic of the overall circuit that solves the problem. 

The operators we use in the construction of our circuit are as follows: 
\begin{itemize}
\item the circular shift operator (ROT)
\item the matching operator (M) for the initialization of the function $\lambda^{i}$
\item the extension operator (EXT$_i$) of the function $\lambda^{i}$
\item the register reversal operator ($\doteqdot$)
\item the controlled bitwise conjunction operator ($\wedge$) and 
\item the register disjunction operator ($\vee$)
\item the copy operator with reverse control (C)
\end{itemize}

\subsection{The Circular Shift Operator}\label{sec:circular_rotation}
A \emph{circular shift operator} (or rotation operator) ROT$_s$ applies a rightward shift of $s$ positions to a register of $n$ qubits so that the element at position $i$ is moved to position $(i+s) \mod n$. In other words, the elements whose position exceeds the size $n$ of the register are moved, in a circular fashion, to the first positions of the register. Such an operator has been effectively used in other quantum text searching algorithms~\cite{Niroula21,CF23}.
The details regarding its construction and its implementation have been detailed by Pavone and Viola in \cite{PV23}, where the authors prove that the resulting operator con be executed in $\mathcal{O}(\log(n))$ time.

More formally, for all $x\in \{0,1\}^n$ and all $s \in \{0,1\}^{\log(n)}$, the controlled circular shift operator performs the following mapping
$$
	\textsc{ROT}_s |x\rangle = |\vv{x}^s\rangle
$$

In our algorithm the cyclic rotation operator will be applied exclusively for rotation values of $2^m$, for $0\leq m\leq \log(n)$. For these types of operators, it was shown in \cite{PV23} that the depth of the resulting circuit is proposrtional to $\log(n)-\log(m)$.
In addition, in our algorithm the application of the cyclic rotation operator will be controlled by a control qubit $\bar{c}$, increasing the cost of that operator by a factor $\mathcal{O}(\log(n))$.
In the graphical representation of a quantum circuit we use the following symbol to identify the cyclic rotation operator controlled by the register $|d\rangle$.

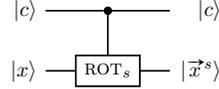
\begin{center}
\begin{tikzpicture}
\node[scale=0.8] {
\begin{quantikz}[row sep={0.5cm,between origins}]
\lstick{\ket{c}} & \ctrl{2} & \qw & & \lstick{\ket{c}} \\
&  &   &  \\
\lstick{\ket{x}} & \gate{\textsc{rot}_s} & \qw  & &  \lstick{\ket{\vv{x}^s}} 
\end{quantikz}
};
\end{tikzpicture}
\end{center}

\subsection{The Operators for Building the Matching Substring Vector}\label{sec:matching_vector}

In a quantum context, each vector $\lambda^i$ is implemented with a quantum register $|\lambda^i\rangle$ of the same size.   The registers are computed iteratively, starting from register $\lambda^0$, and proceedings by computing the register $\lambda^i$ from $\lambda^{i-1}$, for $0< i \leq \log(n)$.

\smallskip 
 
The first register $|\lambda^0\rangle$ is computed by the match operator \textsf{M}. Specifically, the \textsf{M} operator performs, for alla $x,y \in \{0,1\}^n$, the following mapping
$$
	\textrm{M} |x\rangle |y\rangle |0^n\rangle  = |x\rangle |y\rangle |\lambda^0\rangle
$$

Based on the following simple logic relation

$$
	\lambda^0[j]  =  \left( x_j=1 \wedge y_j=1\right) \otimes \left( \neg (x_j=1) \wedge \neg (y_j=1) \right)
$$
the matching gate M is implemented by two sets of $n$ parallel Toffoli gates, the second of which is squeezed by two batteries of parallel X gates.
More formally, to compute the qubit $|\lambda^0[j]\rangle$ we perform the following computation

$$
\begin{array}{ll}
	\textrm{M} |x_i\rangle |y_i\rangle |0\rangle & = (\textrm{X} \otimes \textrm{X} \otimes \textrm{I})\  \textrm{CCX}\  (\textrm{X} \otimes \textrm{X} \otimes \textrm{I})\ \textrm{CCX}\  |x_j\rangle |y_j\rangle |0\rangle \\[0.1cm]
			& = (\textrm{X} \otimes \textrm{X} \otimes \textrm{I})\  \textrm{CCX}\  (\textrm{X} \otimes \textrm{X} \otimes \textrm{I})\  |x_j\rangle |y_j\rangle |x_j \wedge y_j \rangle \\[0.1cm]
			& = (\textrm{X} \otimes \textrm{X} \otimes \textrm{I})\  \textrm{CCX}\  | \neg x_j\rangle | \neg y_j\rangle |x_j \wedge y_j \rangle \\[0.1cm]
			& = (\textrm{X} \otimes \textrm{X} \otimes \textrm{I})\  | \neg x_j\rangle | \neg y_j\rangle | (x_j \wedge y_j) \otimes (\neg x_j \wedge \neg y_j) \rangle \\[0.1cm]
			& = | x_j\rangle | y_j\rangle | (x_j \wedge y_j) \otimes (\neg x_j \wedge \neg y_j) \rangle \\[0.1cm]
			& =  |x_j\rangle |y_j\rangle |\lambda^0[j]\rangle 	
\end{array}
$$

\smallskip

Fig.~\ref{fig:circuit_varphi} (on the left) shows the graphical representation of the \textsf{M} match operator over two input registers $x$ and $y$ of 8 qubits each. The resulting register, $\lambda^0$, of the same size, has the $j$-th qubit set to 1 if and only if $x_j=y_j$. On the right of of the operator is shown its compact graphical form.

The subsequent $|\lambda^i\rangle$ registers, with $i>0$, are instead computed from $|\lambda^{i-1}\rangle$ by the extension operator, named EXT$_i$. Formally, the extension operator EXT$_i$ performs, for all $0< i \leq \log(n)$, the following mapping

$$
	\textrm{EXT}_i |\lambda^{i-1}\rangle |0^n\rangle  = |\lambda^{i-1}\rangle |\lambda^i\rangle 
$$

Such computation is based on the following simple logic relation

$$
	\lambda^i[j] = \lambda^{i-1}[j] \wedge \lambda^{i-1}[j+2^{i-1}]
$$

Therefore the EXT$_i$ operator can be implemented by means of a set of Toffoli gates. More formally we have

$$
\begin{array}{ll}
	\textrm{EXT}_{2^i} |\lambda^{i-1}[j]\rangle  |\lambda^{i-1}[j+2^i]\rangle |0\rangle & = \textrm{CCX}\   |\lambda^{i-1}[j]\rangle  |\lambda^{i-1}[j+2^i]\rangle |0\rangle \\[0.1cm]
			& =  |\lambda^{i-1}[j]\rangle  |\lambda^{i-1}[j+2^i]\rangle |\lambda^{i-1}[j]  \wedge \lambda^{i-1}[j+2^i]\rangle \\[0.1cm]
			& =  |\lambda^{i-1}[j]\rangle  |\lambda^{i-1}[j+2^i]\rangle |\lambda^{i}[j] \rangle 	
\end{array}
$$

It is easy to observe that the extension operator requires constant time, consisting of two sets of parallel Toffoli gates, the first operating on even positions and the second on odd positions. Therefore, it follows that the computation of all registers $|\lambda^i\rangle$, for $0\leq i \leq \log(n)$, requires an overall $\mathcal{O}(\log(n))$ time.

\begin{figure}[!t]
\vspace{-0.5cm}
\begin{center}
\hspace{1cm}$|x\rangle |y\rangle |\lambda^0\rangle \leftarrow$ \textsf{M}$|x\rangle |y\rangle |0^n\rangle$ \hspace{3cm} |\textsc{$\lambda^{i-1}\rangle |\lambda^{i}\rangle \leftarrow$ ext$_{i} |\lambda^{i-1}\rangle |0^n\rangle$}\\[0.2cm] 
\begin{tikzpicture}
\node[scale=0.55] {
\begin{quantikz}[row sep={0.58cm,between origins},column sep = 0.2cm]
&  &   & & & & & & \lstick{\ket{x_0}} & \ctrl{7} & \qw & \qw & \qw & \qw & \qw & \targ{} & \ctrl{7} & \qw & \qw & \qw & \qw & \qw & \targ{} & \qw \\
&  &   & & & & & & \lstick{\ket{x_1}} & \qw & \ctrl{7} & \qw & \qw & \qw & \qw & \targ{} & \qw & \ctrl{7} & \qw & \qw & \qw & \qw & \targ{} & \qw \\
\lstick{\ket{x}} & \ctrl{7} & \qw  & & & & & & \lstick{\ket{x_2}} & \qw & \qw & \ctrl{7} & \qw & \qw & \qw & \targ{} & \qw & \qw & \ctrl{7} & \qw & \qw & \qw & \targ{} & \qw \\
&  &   & & & & & & \lstick{\ket{x_3}} & \qw & \qw & \qw & \ctrl{7} & \qw & \qw & \targ{} & \qw & \qw & \qw & \ctrl{7} & \qw & \qw & \targ{} & \qw \\
&  &   & & & & & & \lstick{\ket{x_4}} & \qw & \qw & \qw & \qw & \ctrl{7} & \qw & \targ{} & \qw & \qw & \qw & \qw & \ctrl{7} & \qw & \targ{} & \qw \\
&  &   & & & & & & \lstick{\ket{x_5}} & \qw & \qw & \qw & \qw & \qw & \ctrl{7} & \targ{} & \qw & \qw & \qw & \qw & \qw & \ctrl{7} & \targ{} & \qw \\
& & & & & & & & & \\
&  &   & & & & & & \lstick{\ket{y_0}} & \ctrl{7} & \qw & \qw & \qw & \qw & \qw & \targ{} & \ctrl{7} & \qw & \qw & \qw & \qw & \qw & \targ{} & \qw \\
&  &   & & & & & & \lstick{\ket{y_1}} & \qw & \ctrl{7} & \qw & \qw & \qw & \qw & \targ{} & \qw & \ctrl{7} & \qw & \qw & \qw & \qw & \targ{} & \qw \\
\lstick{\ket{y}} & \ctrl{7} & \qw  & {\Large =} & & & & & \lstick{\ket{y_2}} & \qw & \qw & \ctrl{7} & \qw & \qw & \qw & \targ{} & \qw & \qw & \ctrl{7} & \qw & \qw & \qw & \targ{} & \qw \\
&  &   & & & & & & \lstick{\ket{y_3}} & \qw & \qw & \qw & \ctrl{7} & \qw & \qw & \targ{} & \qw & \qw & \qw & \ctrl{7} & \qw & \qw & \targ{} & \qw \\
&  &   & & & & & & \lstick{\ket{y_4}} & \qw & \qw & \qw & \qw & \ctrl{7} & \qw & \targ{} & \qw & \qw & \qw & \qw & \ctrl{7} & \qw & \targ{} & \qw \\
&  &   & & & & & & \lstick{\ket{y_5}} & \qw & \qw & \qw & \qw & \qw & \ctrl{7} & \targ{} & \qw & \qw & \qw & \qw & \qw & \ctrl{7} & \targ{} & \qw \\
& & & & & & \\
&  &   & & & & & & \lstick{\ket{\lambda^0_0}} & \targ{} & \qw & \qw & \qw & \qw & \qw & \qw & \targ{}  & \qw & \qw & \qw & \qw & \qw & \qw & \qw \\
&  &  & & & & & & \lstick{\ket{\lambda^0_1}} & \qw & \targ{} & \qw & \qw & \qw & \qw & \qw & \qw & \targ{}  & \qw & \qw & \qw & \qw & \qw & \qw \\
\lstick{\ket{\lambda^0}} &  \gate{M} & \qw & & & & & & \lstick{\ket{\lambda^0_2}} & \qw & \qw & \targ{}  & \qw & \qw & \qw & \qw & \qw & \qw & \targ{}  & \qw & \qw & \qw & \qw & \qw \\
&  &   & & & & & & \lstick{\ket{\lambda^0_3}} & \qw & \qw & \qw & \targ{} & \qw & \qw & \qw & \qw & \qw & \qw & \targ{}  & \qw & \qw & \qw & \qw \\
&  &   & & & & & & \lstick{\ket{\lambda^0_4}} & \qw & \qw & \qw & \qw & \targ{} & \qw & \qw & \qw & \qw & \qw & \qw & \targ{}  & \qw & \qw & \qw \\
&  &   & & & & & & \lstick{\ket{\lambda^0_5}} & \qw & \qw & \qw & \qw & \qw & \targ{} & \qw & \qw & \qw & \qw & \qw & \qw & \targ{}  & \qw & \qw
\end{quantikz}\hspace{1cm}
\begin{quantikz}[row sep={0.5cm,between origins}]
&  &   & & &\lstick{\ket{\lambda^0_0}} & \ctrl{1} & \qw & \qw & \qw & \qw & \qw \\
&  &   & & &\lstick{\ket{\lambda^0_1}} & \ctrl{6} & \qw & \qw & \ctrl{1} & \qw & \qw \\
\lstick{\ket{\lambda^0}} & \octrl{7} & \qw  & & &\lstick{\ket{\lambda^0_2}} & \qw & \ctrl{1} & \qw & \ctrl{6} & \qw & \qw \\
&  &   & & &\lstick{\ket{\lambda^0_3}} & \qw & \ctrl{6} & \qw & \qw & \ctrl{1} & \qw \\
&  &   & & &\lstick{\ket{\lambda^0_4}} & \qw & \qw & \ctrl{1} & \qw & \ctrl{6} & \qw \\
&  &   & & &\lstick{\ket{\lambda^0_5}} & \qw & \qw & \ctrl{6} & \qw & \qw & \qw \\
&  &   & {\Large =} & && & & & & \\
&  &   & & &\lstick{\ket{\lambda^1_0}} & \targ{} & \qw & \qw & \qw & \qw & \qw \\
&  &   & & &\lstick{\ket{\lambda^1_1}} & \qw & \qw{} & \qw & \targ{} & \qw & \qw \\
\lstick{\ket{\lambda^1}} & \gate{\textsc{ext}_1} & \qw  & & &\lstick{\ket{\lambda^1_2}} & \qw & \targ{} & \qw & \qw & \qw & \qw \\
&  &   & & &\lstick{\ket{\lambda^1_3}} & \qw & \qw & \qw & \qw & \targ{} & \qw \\
&  &   & & &\lstick{\ket{\lambda^1_4}} & \qw & \qw & \targ{} & \qw & \qw & \qw \\
& & & & & \\
& & & & & \\
&  &   & & &\lstick{\ket{\lambda^1_0}} & \ctrl{2} & \qw & \qw & \qw \\
&  &   & & &\lstick{\ket{\lambda^1_1}} & \qw & \ctrl{2} & \qw & \qw \\
\lstick{\ket{\lambda^1}} & \octrl{5}  & \qw  & & &\lstick{\ket{\lambda^1_2}} & \ctrl{4} & \qw & \ctrl{2} & \qw \\
&  &   & & &\lstick{\ket{\lambda^1_3}} & \qw & \ctrl{4} & \qw & \qw \\
&  &   & & &\lstick{\ket{\lambda^1_4}} & \qw & \qw & \ctrl{4} & \qw \\
&  &   & {\Large =} & && & & &\\
&  &   & & &\lstick{\ket{\lambda^2_0}} & \targ{} & \qw & \qw & \qw \\
\lstick{\ket{\lambda^2}} & \gate{\textsc{ext}_2} &  \qw & & &\lstick{\ket{\lambda^2_1}} & \qw & \targ{} & \qw & \qw \\
&  &   & & &\lstick{\ket{\lambda^2_2}} & \qw & \qw & \targ{} & \qw 
\end{quantikz}
};
\end{tikzpicture}
\end{center}
\vspace{-0.5cm}
\caption{\label{fig:circuit_varphi}(On the left) The graphical representation of the \textsf{M} match operator over two input registers $x$ and $y$ of 8 qubits each. The resulting register, $\lambda^0$, of the same size, has the $j$-th qubit set to 1 if and only if $x_j=y_j$. (On the right) The graphical representation of two extension operators, used to compute the vector $\lambda^i$ from the vector $\lambda^{i-1}$. The $j$-th qubit of $\lambda^{i}$ is set to 1 if and only if the two qubits at positions $j$ and $j+i$ have both value $1$ in $\lambda^{i-1}$. On the right of each of the operators is shown its compact graphical form.}
\end{figure}
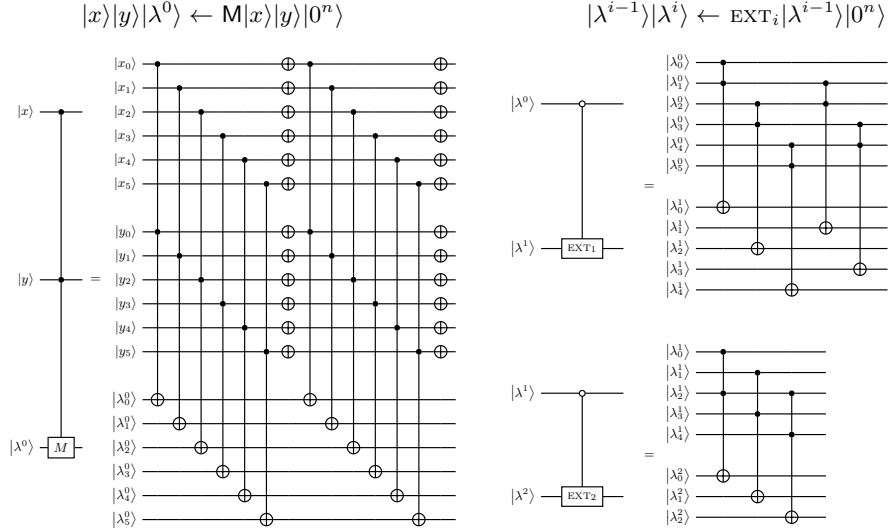

Fig.~\ref{fig:circuit_varphi} (on the right) shows the graphical representation of two extension operators, used to compute the vector $\lambda^i$ from the vector $\lambda^{i-1}$. The $j$-th qubit of $\lambda^{i}$ is set to 1 if and only if the two qubits at positions $j$ and $j+i$ have both value $1$ in $\lambda^{i-1}$. On the right of each of the operators is shown its compact graphical form.

\subsection{The Register Reversal Operators}\label{sec:reversal}

The register reversal operator, which we indicate by symbol $\doteqdot$, operates on a single register $|q\rangle$, of size $n$, by inverting the positions of its qubits. More formally, for all $q \in \{0,1\}^n$, the register reversal operator performs the following mapping

$$
	\doteqdot |q\rangle = |\bar{q}\rangle
$$

In terms of linear algebra such an operator can be obtained as follows

$$
	\displaystyle \doteqdot |q\rangle = \bigotimes_{i=0}^{\lfloor n/2 \rfloor} \Big( \textsc{Swap} |q_i\rangle |q_{n-i-1}\rangle \Big)
$$

Although such an operation requires the application of $\lfloor n/2 \rfloor$ swaps between pairs of qubits, these swap operations can be performed in parallel, requiring only constant time for the inversion of the entire register. 

Within our algorithm we make use of a register reversal operator in the initialization phase of the parameter $d$ that identifies the length of the sub-strings to be compared within $x$ and $y$. Specifically, we doctored $\bar{d}$ from the input parameter $d$ through the mapping $\doteqdot |d\rangle \rightarrow |\bar{d}\rangle$.

\subsection{The Controlled Bitwise Conjunction Operator}\label{sec:conjunction}

Finally, we present in this section the two circuits to implement the bitwise conjunction operator ($\wedge$) and the disjunction operator ($\vee$).

The bitwise conjunction operator performs bit-to-bit logical \textsc{and}  between two registers, $|a\rangle$ and $|b\rangle$, of $n$ qubits, depositing the result within a third register $|q\rangle$, whose size is still $n$. However, we are interested in a controlled version of such an operator, in which a fourth register $|c\rangle$, formed by a single qubit plays the role of control qubit. In other words, the bit-to-bit logical \textsc{and} is applied only in the case where the control qubit contains the value $0$. Nothing will be done otherwise.

Formally, the controlled bitwise conjunction operator performs, for all $a,b\in \{0,1\}^n$, the following mapping

$$
	\displaystyle \wedge |c\rangle |a\rangle |b\rangle |0^n\rangle = |c\rangle |a\rangle |b\rangle  \bigotimes_{i=0}^{n-1} | c_i \wedge a_i \wedge b_i \rangle
$$

Such an operator can be implemented through the application of $n$ gate fanouts, each of which consists of 3 controls, $|c\rangle$, $|a_i\rangle$, and $|b_i\rangle$, and has the register $|q_i\rangle$ as the target, for $0\leq i < n$.
Although these $n$ fanout operators have a common control qubit ($|c\rangle$) that does not allow their parallel application, it is possible to apply a well-known technique that copies the value of $|c\rangle$ on $n$ ancillae qubits before applying the $n$ fanout operators in parallel.
This technique then allows us to construct a circuit for this operator with depth $\mathcal{O}(\log(n))$.

\begin{figure}[!t]
\begin{center}
\begin{tikzpicture}
\node[scale=0.55] {
\begin{quantikz}[row sep={0.5cm,between origins}]
\lstick{\ket{c}} & \ctrl{4} & \qw  & & &\lstick{\ket{c}} & \ctrl{2} & \ctrl{3} & \ctrl{4} & \ctrl{5} & \ctrl{6} & \ctrl{7} & \qw \\
&  &   &  & & & & & & & \\

&  &   & & &\lstick{\ket{a_0}} & \ctrl{7} & \qw & \qw & \qw & \qw & \qw & \qw \\
&  &   & & &\lstick{\ket{a_1}} & \qw & \ctrl{7} & \qw & \qw & \qw & \qw & \qw \\
\lstick{\ket{a}} & \octrl{7} & \qw  & & &\lstick{\ket{a_2}} & \qw & \qw & \ctrl{7} & \qw & \qw & \qw & \qw \\
&  &   & & &\lstick{\ket{a_3}} & \qw & \qw & \qw & \ctrl{7} & \qw & \qw & \qw \\
&  &   &  & &\lstick{\ket{a_4}} & \qw & \qw & \qw & \qw & \ctrl{7} & \qw & \qw \\
&  &   & & &\lstick{\ket{a_5}} & \qw & \qw & \qw & \qw & \qw & \ctrl{7} & \qw \\
&  &   & {\Large =} & & & & & & & \\

&  &   & & &\lstick{\ket{b_0}} & \ctrl{7} & \qw & \qw & \qw & \qw & \qw & \qw \\
&  &   & & &\lstick{\ket{b_1}} & \qw & \ctrl{7} & \qw & \qw & \qw & \qw & \qw \\
\lstick{\ket{b}} & \octrl{7} & \qw  & & &\lstick{\ket{b_2}} & \qw & \qw & \ctrl{7} & \qw & \qw & \qw & \qw \\
&  &   & & &\lstick{\ket{b_3}} & \qw & \qw & \qw & \ctrl{7} & \qw & \qw & \qw \\
&  &   &  & &\lstick{\ket{b_4}} & \qw & \qw & \qw & \qw & \ctrl{7} & \qw & \qw \\
&  &   & & &\lstick{\ket{b_5}} & \qw & \qw & \qw & \qw & \qw & \ctrl{7} & \qw \\

&  &   &  & & & & & & & \\
&  &   & & &\lstick{\ket{q_0}} & \targ{} & \qw & \qw & \qw & \qw & \qw & \qw \\
&  &   & & &\lstick{\ket{q_1}} & \qw & \targ{} & \qw & \qw & \qw & \qw & \qw \\
\lstick{\ket{q}} & \gate{\wedge} & \qw  & & &\lstick{\ket{q_2}} & \qw & \qw & \targ{} & \qw & \qw & \qw & \qw \\
&  &   & & &\lstick{\ket{q_3}} & \qw & \qw & \qw & \targ{} & \qw & \qw & \qw \\
&  &   &  & &\lstick{\ket{q_4}} & \qw & \qw & \qw & \qw & \targ{} & \qw & \qw \\
&  &   & & &\lstick{\ket{q_5}} & \qw & \qw & \qw & \qw & \qw & \targ{} & \qw \\
\end{quantikz}~~~~~
\begin{quantikz}[row sep={0.5cm,between origins}]
\lstick{\ket{c}} & \octrl{4} & \qw  & & &\lstick{\ket{c}} & \targ{} & \ctrl{2} & \ctrl{3} & \ctrl{4} & \ctrl{5} & \ctrl{6} & \ctrl{7} & \targ{} & \qw \\
&  &   &  & & & & & & & \\

&  &   & & &\lstick{\ket{a_0}} & \qw & \ctrl{7} & \qw & \qw & \qw & \qw & \qw & \qw & \qw \\
&  &   & & &\lstick{\ket{a_1}} & \qw & \qw & \ctrl{7} & \qw & \qw & \qw & \qw & \qw & \qw \\
\lstick{\ket{a}} & \ctrl{7} & \qw  & & &\lstick{\ket{a_2}} & \qw & \qw & \qw & \ctrl{7} & \qw & \qw & \qw & \qw & \qw \\
&  &   & & &\lstick{\ket{a_3}} & \qw & \qw & \qw & \qw & \ctrl{7} & \qw & \qw & \qw & \qw \\
&  &   &  & &\lstick{\ket{a_4}} & \qw & \qw & \qw & \qw & \qw & \ctrl{7} & \qw & \qw & \qw \\
&  &   & & &\lstick{\ket{a_5}} & \qw & \qw & \qw & \qw & \qw & \qw & \ctrl{7} & \qw & \qw \\
&  &   &  {\Large =} & & & & & & & \\

&  &   & & &\lstick{\ket{b_0}} & \qw & \targ{} & \qw & \qw & \qw & \qw & \qw & \qw & \qw \\
&  &   & & &\lstick{\ket{b_1}} & \qw & \qw & \targ{} & \qw & \qw & \qw & \qw & \qw & \qw \\
\lstick{\ket{b}} & \gate{\textsc{c}} & \qw  & & &\lstick{\ket{b_2}} & \qw & \qw & \qw & \targ{} & \qw & \qw & \qw & \qw & \qw \\
&  &   & & &\lstick{\ket{b_3}} & \qw & \qw & \qw & \qw & \targ{} & \qw & \qw & \qw & \qw \\
&  &   & & &\lstick{\ket{b_4}} & \qw & \qw & \qw & \qw & \qw & \targ{} & \qw & \qw & \qw \\
&  &   & & &\lstick{\ket{b_5}} & \qw & \qw & \qw & \qw & \qw & \qw & \targ{} & \qw & \qw 
\end{quantikz}
};
\end{tikzpicture}
\end{center}
\vspace{-0.5cm}
\caption{\label{fig:circuit_copy}The graphical representation of the copy operator with reversal control over one input registers $|a\rangle$ of 6 qubits. The resulting qubit, $|b\rangle$, is a copy of $|a\rangle$ only if the control qubit $|c\rangle$ is set to $0$. On the right of each of the operator is shown its compact graphical form.}
\end{figure}
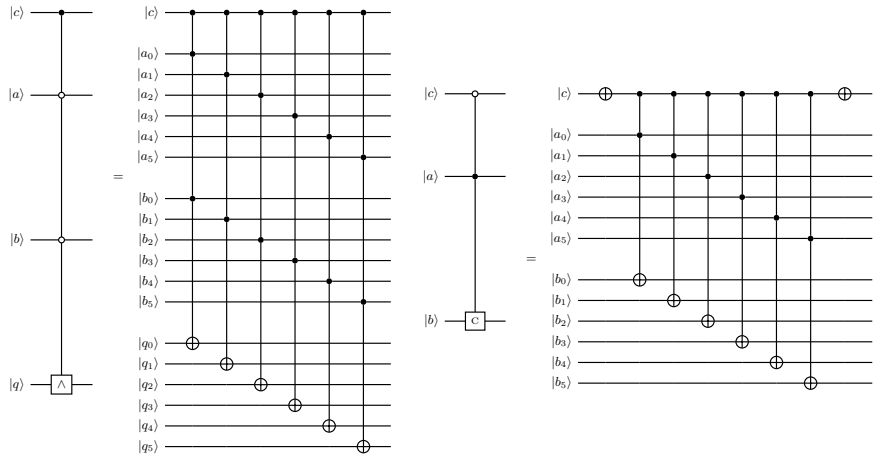

\subsection{The Copy Operator with Reversal Control}\label{sec:disjunction}
The copy operator with reversal control (or simply CRC operator), is a simple copy operator, whose task is to transfer the value of the $n$ qubits of a first $|a\rangle$ register into the $n$ qubits of a second $|b\rangle$ register, initialized to $|0^n\rangle$. In addition we envisage the presence of an inverse control $|c\rangle$ qubit for the application of this operation. In this context, an inverse control allows the application of the operation when the value of the control qubit is $0$, while inhibiting its application when the value of the qubit is 1.

Formally, the CRC operator performs, for all $a\in \{0,1\}^n$ and $c\in \{0,1\}$,  the following mapping

$$
	\displaystyle \textrm{C} |c\rangle |a\rangle |0^n\rangle = |c\rangle |a\rangle  \bigotimes_{i=0}^{n-1} \left\vert (\neg c_i) \wedge a_i \right\rangle
$$

The implementation of the CRC operator goes through the use of $n$ Toffoli gates, where the $i$-th Toffoli gate, for $0\leq i <n$, uses $|c\rangle$ and $|a_i\rangle$ as controls and $|b_i\rangle$ as targets. 
Although these $n$ Toffoli gates have a common control qubit ($|c\rangle$) that does not allow their parallel application, it is possible to apply a well-known technique that copies the value of $|c\rangle$ on $n$ ancillae qubits before applying the $n$ fanout operators in parallel.
This technique then allows us to construct a circuit for this operator with depth $\mathcal{O}(\log(n))$.

\subsection{The Register Disjunction Operator}\label{sec:disjunction}
The register disjunction operator, or simply disjunction operator, computes the logic \textsc{or} operation between the $n$ qubits of the input register $|a\rangle$ and deposits the result of such logic operation inside the output qubit, $|r\rangle$.

Formally, the register disjunction operator performs, for all $a\in \{0,1\}^n$, the following mapping

$$
	\displaystyle \vee |a\rangle |0\rangle = |a\rangle \left\vert \bigvee_{i=0}^{n-1} a_i \right\rangle
$$

The implementation of the disjunction operator can be realized by a fanout operator on the $|a\rangle$ registers, and with target $|r\rangle$, squeezed by two batteries of X operators on the $|a\rangle$ qubits. The depth of the resulting circuit is thus equal to $\mathcal{O}(\log(n))$.

\begin{figure}[!t]
\begin{center}
\begin{tikzpicture}
\node[scale=0.6] {
\begin{quantikz}[row sep={0.5cm,between origins}]
&  &   & & &\lstick{\ket{a_0}} & \targ{x} & \ctrl{1} & \targ{x} & \qw \\
&  &   & & &\lstick{\ket{a_1}} & \targ{x} & \ctrl{1} & \targ{x} & \qw \\
\lstick{\ket{a}} & \ctrl{5} & \qw  & & &\lstick{\ket{a_2}} & \targ{x} & \ctrl{1} & \targ{x} & \qw \\
&  &   & & &\lstick{\ket{a_3}} & \targ{x} & \ctrl{1} & \targ{x} & \qw \\
&  &   & {\Large =} & &\lstick{\ket{a_4}} & \targ{x} & \ctrl{1} & \targ{x} & \qw \\
&  &   & & &\lstick{\ket{a_5}} & \targ{x} & \ctrl{2} & \targ{x} & \qw \\
&  &   &  & & & & & & & \\
\lstick{\ket{r}} & \gate{\vee} & \qw  & & &\lstick{\ket{r}} & \qw & \targ{} & \targ{} & \qw  
\end{quantikz}
};
\end{tikzpicture}
\end{center}
\vspace{-0.5cm}
\caption{\label{fig:circuit_or}The graphical representation of the register disjunction operator over two input registers $|a\rangle$ of 6 qubits each. The resulting qubit, $|r\rangle$, is set to $1$ if at least one of the qubits of $|a\rangle$ is worth 1. On the right of each of the operator is shown its compact graphical form.}
\end{figure}
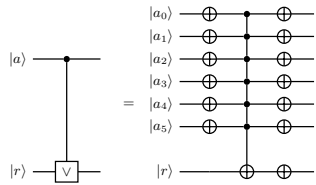

\subsection{Quantum Circuit Definition for FSM Problems}

We are now ready to define quantum circuits for FSM problem solving. As anticipated in the previous sections, the circuit that solves the three proposed problems is the same: what changes is the way the input registers are initialized.

Specifically, we will show how the three problems can be solved through a quantum circuit with depth equal to $\mathcal{O}(\log^3(n))$. 
The complete circuit we describe in this section is depicted in Fig.\ref{fig:circuit_bms}.

The circuit make use of two registers $|x\rangle$ and $|y\rangle$, both of size $n$ which we assume to contain the characters of the two input strings $x$ and $y$, respectively. We assume also that these registers are stored in a quantum memory and do not need initialization.
The register containing the value of the fixed length $d$ is stored in a $|d\rangle$ register of size $\log(d)$. Initialization of that register would take time $\mathcal{O}(\log(n))$. For simplicity, in Fig.\ref{fig:circuit_bms} we assume a value of $d$, with $1\leq d <8$, which can represented by $3$ bits.
The output of the computation is stored in the $|r\rangle$ register consisting of a single qubit.

Observe that the initialization of quantum registers also involves applying the register reversal operator to $|d\rangle$ so that it can hold the value $\bar{d}$. This operation is performed in constant time (see Section \ref{sec:reversal}).

In our quantum implementation, each vector $D^i$, with $0\leq i<\log(d)$, is maintained in a quantum register $|D^i\rangle$ of size $n+1$, while each vector $\lambda^i$, is maintained in a quantum register $|\lambda^i\rangle$ of size $n$.

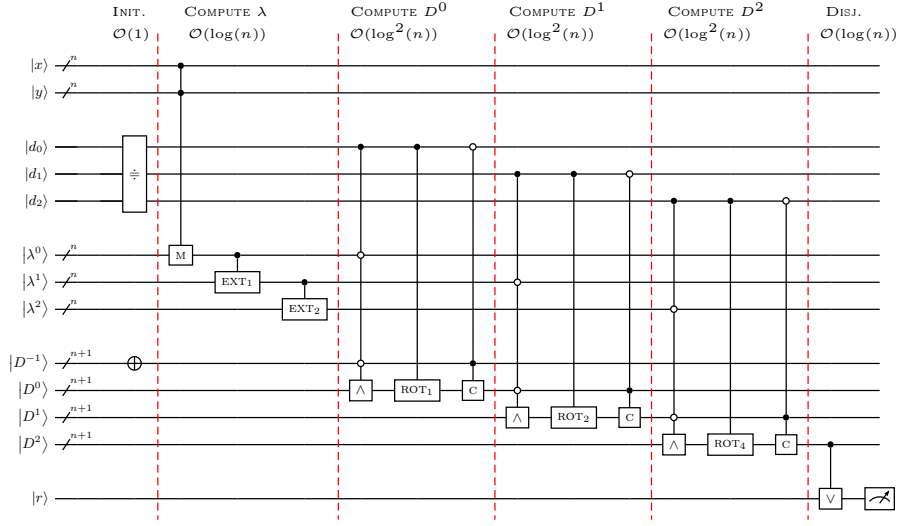
\begin{figure}[!t]
{\color{white}.}{\tiny \hspace{1.5cm}\textsc{Init.}\hspace{0.5cm}\textsc{Compute $\lambda$}\hspace{1.1cm}\textsc{Compute $D^0$}\hspace{0.8cm}\textsc{Compute $D^1$}\hspace{0.8cm}\textsc{Compute $D^2$}\hspace{0.8cm}\textsc{Disj.}}\\[-0.1cm]
{\color{white}.}{\tiny \hspace{1.5cm}$\mathcal{O}(1)$\hspace{0.5cm}$\mathcal{O}(\log(n))$\hspace{1.1cm}$\mathcal{O}(\log^2(n))$\hspace{0.9cm}$\mathcal{O}(\log^2(n))$\hspace{0.9cm}$\mathcal{O}(\log^2(n))$\hspace{0.9cm}$\mathcal{O}(\log(n))$}\\
\vspace{-1cm}
\begin{center}
\begin{tikzpicture}
\node[scale=0.60] {
\begin{quantikz}[row sep={0.6cm,between origins}]
& & & \slice{} & & & \slice{} & & & \slice{} & & & \slice{} & & & \slice{} \\

\lstick{\ket{x}} & \qwbundle{n} & \qw & \qw & \ctrl{1} & \qw & \qw 							& \qw & \qw & \qw & \qw & \qw & \qw & \qw & \qw & \qw & \qw & \qw \\

\lstick{\ket{y}} & \qwbundle{n} & \qw & \qw  & \ctrl{6} & \qw & \qw 							& \qw & \qw & \qw & \qw & \qw & \qw & \qw & \qw & \qw & \qw & \qw \\

& & & &&&&&&&&&&&&& \\

\lstick{\ket{d_0}} & \qw \qw  & \qw & \gate[3]{\doteqdot} & \qw & \qw & \qw 			& \ctrl{4} & \ctrl{9} & \octrl{8} & \qw & \qw & \qw & \qw & \qw & \qw & \qw & \qw \\			
\lstick{\ket{d_1}} & \qw \qw  & \qw & \qw & \qw & \qw & \qw			& \qw & \qw & \qw & \ctrl{4} & \ctrl{9} & \octrl{8} & \qw & \qw & \qw  & \qw & \qw \\			
\lstick{\ket{d_2}} & \qw \qw  & \qw & \qw & \qw & \qw & \qw			& \qw & \qw & \qw & \qw & \qw & \qw & \ctrl{4} & \ctrl{9} & \octrl{8} & \qw & \qw \\

& & & &&&&&&&&&&&&& \\

\lstick{\ket{\lambda^0}} & \qwbundle{n} & \qw & \qw  & \gate{\textsc{m}} & \ctrl{1} & \qw			& \octrl{4} & \qw & \qw & \qw & \qw  & \qw & \qw & \qw & \qw & \qw & \qw \\

\lstick{\ket{\lambda^1}} & \qwbundle{n} & \qw & \qw  & \qw & \gate{\textsc{ext}_1} & \ctrl{1}  		& \qw & \qw & \qw & \octrl{4} & \qw & \qw & \qw & \qw & \qw & \qw & \qw \\

\lstick{\ket{\lambda^2}} & \qwbundle{n} & \qw & \qw  & \qw & \qw & \gate{\textsc{ext}_2}		& \qw & \qw & \qw & \qw & \qw & \qw & \octrl{4} & \qw & \qw & \qw & \qw \\

& & & &&&&&&&&&&&&& \\

\lstick{\ket{D^{-1}}} & \qwbundle{n+1} & \qw & \targ{} & \qw & \qw & \qw 						& \octrl{1} & \qw & \ctrl{1} & \qw & \qw & \qw & \qw & \qw & \qw & \qw & \qw \\

\lstick{\ket{D^0}} & \qwbundle{n+1} & \qw & \qw & \qw & \qw & \qw	 						& \gate{\wedge} & \gate{\textsc{rot}_1} & \gate{\textsc{c}} & \octrl{1} & \qw & \ctrl{1} & \qw & \qw & \qw  & \qw & \qw \\

\lstick{\ket{D^1}} & \qwbundle{n+1} & \qw & \qw & \qw & \qw & \qw 						& \qw & \qw & \qw & \gate{\wedge} & \gate{\textsc{rot}_2} & \gate{\textsc{c}} & \octrl{1} & \qw & \ctrl{1} & \qw & \qw \\

\lstick{\ket{D^2}} & \qwbundle{n+1} & \qw & \qw & \qw & \qw & \qw 						& \qw & \qw & \qw & \qw & \qw & \qw & \gate{\wedge} & \gate{\textsc{rot}_4} & \gate{\textsc{c}} & \ctrl{2} & \qw  \\

& & & &&&&&&&&&&&&& \\

\lstick{\ket{r}} & \qw & \qw & \qw & \qw & \qw & \qw & \qw & \qw & \qw & \qw & \qw & \qw & \qw & \qw & \qw & \gate{\vee} & \meter{}
\end{quantikz}
};
\end{tikzpicture}
\end{center}
\vspace{-0.5cm}
\caption{\label{fig:circuit_bms}A simplified schema of the circuit used for the \textsc{FSM} algorithm. 
}
\end{figure}

The registers holding the matching substring vectors $\lambda^i$ are computed in the first phase of the algorithm through the application of the matching operator M (for the initialization of $\lambda^0$) and the extension operators EXT$_i$ for subsequent iterations. Since the application of these operators takes constant time (see Section \ref{sec:matching_vector}) the whole process can be executed in time $\mathcal{O}(\log(n))$.

The next stage of computation constructs the $D^i$ vectors, for $0\leq i \log(d)$, through $\log(d)$ iterative steps all of which have the same structure. At the $i$-th step, the bitwise conjunction operator between the $\lambda^i$ and $D^{i-1}$ registers, controlled by the $|d_i\rangle$ register, is applied in order to implement the logical relation defined in (\ref{eq:varphi}). Then a rightward rotation is applied to the operator $D^i$ of $2^i$ positions, still controlled by the register $|d_i\rangle$. In the event that the register $|d_i\rangle$ should contain the value $0$ then the application of the copy operator with inverse control, C, transfers the register $D^{i-1}$ to $D^i$. These three operators take time $\mathcal{O}(\log(n))$, $\mathcal{O}(\log^2(n))$ and $\mathcal{O}(\log(n))$, respectively. So the entire computation of registers $D^i$ can be done in time $\mathcal{O}(\log^3(n))$.

The last step in the circuit is to apply the register disjunction operator on $|D^{\log(d)}\rangle$, so as to implement the logical operation defined in Fig.\ref{fig:algorithm-bpm-pseudocode} (line 7).
We therefore conclude that the total complexity for executing the FPM circuit is $\mathcal{O}(\log^3(n))$. 

\begin{observation}
As stated in our Observation \ref{obs:binary}, in the course of this paper we assumed that the input strings, $x$ and $y$, were formed by characters taken from a binary alphabet. In the case where the two strings were instead composed of characters drawn from an alphabet $\Sigma$ of size $\sigma>2$, we can assume that each of the characters is representable by a number of bits equal to $\min(\log(\sigma),\log(n))$.
All the gates presented in this paper can thus be adapted to deal with nonbinary strings, incurring an overhead of complexity at most equal to $\mathcal{O}(\log(n))$. This allows us to conclude that the complexity of circuits solving FSM problems in the general case is equal to $\mathcal{O}(\log^4(n))$.
\end{observation}

\newpage

\bibliographystyle{plain}


\end{document}